\documentclass[a4paper,11pt]{article}
\usepackage{amsmath,amssymb,amsfonts,amsthm}
\usepackage[english]{babel}
\usepackage{mathtools}
\usepackage{url}

\usepackage{url}

\usepackage{graphicx}

\newtheorem{corollary}{Corollary}[section] 
\newtheorem{theorem}{Theorem}[section]
\newtheorem{lemma}{Lemma}[section]
\newtheorem{remark}{Remark}[section]

\newtheorem{definition}{Definition}[section]

\title{The Ising model on the random planar causal triangulation: bounds on the critical line and magnetization properties}
\author{George M. Napolitano and Tatyana S. Turova \\ \small{\textit {Centre for Mathematical Sciences}} \\ \small{\textit{  Division of Mathematical Statistics, Lund University}} \\ \small{\textit{S\"olvegatan 18, 22100, Lund, Sweden}}}
\date{}

\begin{document}

\maketitle

\begin{abstract}
We investigate a Gibbs (annealed) probability measure defined on Ising spin configurations on causal triangulations of the plane. We study the region where such measure can be defined and provide bounds on the boundary of this region (critical line). We prove that for any finite random triangulation the magnetization of the central spin is sensitive to the boundary conditions. Furthermore, we show that in the infinite volume limit, the magnetization of the central spin vanishes for values of the temperature high enough.
\end{abstract}

\section{Introduction}

In the last few decades there has been an increasing interest by the scientific community in random graphs, mostly due to their wide range of application in several branches of science.

In theoretical physics, random graphs have been used as a tool to set up discrete models of quantum gravity. The basic idea underlying such models is the discretization of spacetime \textit{via} triangulations (in the spirit of Regge calculus) and the representation of fluctuating geometries, naturally arising in the path-integral approach to the quantization of gravity, in terms of random triangulations. 

More precisely, according to the path-integral formalism, the partition function in (Euclidean) quantum gravity is formally given by the expression
\begin{equation}
\label{eq:qg_partfun}
	\mathcal{Z}(\mu) = \int \mathcal{D}[g] \; e^{- \mathcal{S}_{\mu}[g]}
\end{equation}
where $\mathcal{S}_{\mu}[g]$ is the Einstein-Hilbert action, $\mu$ is the cosmological constant and the integral is intended over the space of metrics (modulo diffeomorphisms) of a manifold $M$.
A discretization scheme is often used to make sense of such expression. This can be implemented, for example, by triangulating the manifold $M$ by finite triangulations built up of equilateral triangles of side $a$, so that each metric (the dynamical variables of the continuum theory) corresponds to a different triangulation. In such framework, the discrete counterpart of the partition function (\ref{eq:qg_partfun}) can be defined as
\begin{equation}
	Z(\mu) = \sum_{T \in \mathcal{T}_M} e^{- S_\mu[T]},
\end{equation}
where $\mathcal{T}_M$ is the set of all inequivalent triangulations of $M$ and $S_\mu[T]$ is the discrete analog of the Einstein-Hilbert action. In two dimensions $S_\mu[T]$ consists of the volume term alone, thus the discrete partition function reads
\begin{equation}
	Z(\mu) = \sum_{T \in \mathcal{T}_M} e^{- \mu |F(T)|},
\end{equation}
where $|F(T)|$ denotes the number of triangles in $T$ and a factor proportional to $a^2$ has been absorbed in $\mu$. Note, however, that the above discussion is purely heuristic, as a rigorous mathematical treatment of the gravitational path-integral and its discretization is still missing. We refer the reader to  \cite{ADJ-qg} for a comprehensive treatment of this subject.

In \cite{ADF85,FD85,KKM85}, the so-called  \textit{dynamical triangulation} model is introduced as a triangulation technique of Euclidean surfaces. It was later discovered that such model produces some non-physical solutions, namely the presence of causality violating geometries.

The \textit{causal dynamical triangulation} (CDT) model was first proposed in \cite{AL98} as a possible cure for such anomalies. In this model a causal structure (mimicking that of Minkowski spacetime) is introduced in the theory from the start, by restricting the class of allowed triangulations to those that can be sliced perpendicularly to the time direction and with fixed topology on the spatial slices.

Malyshev \cite{Malyshev01} gave a solid  mathematical ground  for such models, developing a  general theory of 
 Gibbs fields on random spaces where matter (represented by the configurations of spins) is naturally coupled with the gravity (represented by the graph). 

Observe that even without spins the CDT model for two-dimensional surfaces exhibits a non-trivial phase transition. 
This model  is solved analytically: the  partition function is explicitly derived along with the scaling limits for the correlation functions in \cite{MYZ01}. 
 Nowadays most of its geometrical properties, such as its Hausdorff and spectral dimension \cite{DJW3},  are very well understood. In particular, it has been shown \cite{DJW3,MYZ01} that the model exhibits a non-trivial behaviour: 
depending on the parameters  the limiting average surface can behave as one-dimensional (subcritical regime), whereas at a certain value (criticality) it has 
properties of two-dimensional space. 

In such a framework, it is certainly interesting to consider statistical mechanical models on random planar graphs, as they can be seen  as the discrete realization of the coupling between matter fields and gravity. Probably, one of the most well-known of these systems is the Ising model on planar random lattice. This was studied and exactly solved by Kazakov et al. in \cite{kazakov,BouKaz,BDKS}, using matrix model techniques. However, the Ising model on causal triangulation seems to be a much more difficult problem to address. 

The Ising model on the critical random causal triangulation fixed {\it a priori} (the so-called \textit{quenched} version) has been proved to exhibit a phase transition \cite{KY}. Also, for the  \textit{annealed} coupling (Ising-spin configuration and triangulation sampled together at random) some  numerical \cite{AAL} and even analytic results  \cite{HSYZ}  are obtained.
Still the exact solution to the Ising model on CDT remains to be an open question.

It is worth mentioning that in recent years statistical models have been studied on different random geometries, other than  random causal triangulations, such as random trees \cite{DN} and higher dimensional random graphs \cite{BGR12}. In particular, in \cite{DN} the authors study the annealed coupling of an Ising model (with external magnetic field) with a certain class of random trees (called \textit{generic} trees). It has been proved that the set of infinite trees decorated with Ising spins consists of trees made of a single infinite path stemming from the tree's root, and on whose vertices finite trees are attached. This shows, by comparison with the case without spins \cite{DJW2}, that the spin system does not affect the geometry of the underlying graph. Furthermore, it is proved that, as a consequence of this 1-dimensional feature,  the infinite spin system does not experience any phase transition.

Quenched models of Ising spins on random graphs have also been studied in the recent literature, for example in \cite{DM10}, where Ising models on certain types of random graphs (locally tree-like graphs) are considered and an explicit expression of the free energy is provided.

Here we study the annealed coupling of an Ising model, at inverse temperature $\beta$, on a random planar causal triangulation. We provide bounds on the region of the $(\beta,\mu)$-plane where the partition function exists, improving the bounds known so far \cite{HSYZ}, at least in the low- and high-temperature regions. Furthermore, we discuss some magnetization properties of the spin system, proving that at sufficiently high-temperature, the mean magnetization of the central spin goes to zero in the thermodynamic limit.

This paper is organized as follows. In Sec. \ref{sec:defs}, we define
 the model and, following \cite{Malyshev01}, introduce the concept of \textit{spin-graph}. In Sec. \ref{sec:results} we collect the main results and we discuss them.  All the proofs are given in the remaining sections. In the Appendix we discuss the analytic properties of a series which is used throughout the paper.

\section{Definition of the model}
\label{sec:defs}

\subsection{Causal triangulations with spins}

\begin{definition}\label{DM}
A \textit{causal triangulation} $T$ is a rooted planar locally finite connected graph satisfying the following properties.
\begin{enumerate}
	\item The set of vertices at graph distance $i$ from the root vertex, together with the edges connecting them, form a cycle, denoted by $S_i=S_i(T)$ (when there is only one vertex the corresponding cycle has only one edge, i.e. it is a loop).
	\item All faces of the graph are triangles, with the only exception of the external face.
	\item One edge attached to the root vertex is marked, we call it \textit{root edge}. 
\end{enumerate}
\end{definition}

The last condition in the above definition is a technical requirement, needed to cancel out possible rotational symmetries around the root.
Here a triangle is defined as a face with exactly 3 edges incident to it, with the convention that an edge incident to the same face on both sides counts twice. 
We shall call $S_i$ the $i$-th \textit{slice} of the triangulation $T$. An example of such triangulation is showed in Fig. \ref{fig:ct}.

\begin{figure}[t!]%
\centering
\includegraphics[scale=0.8]{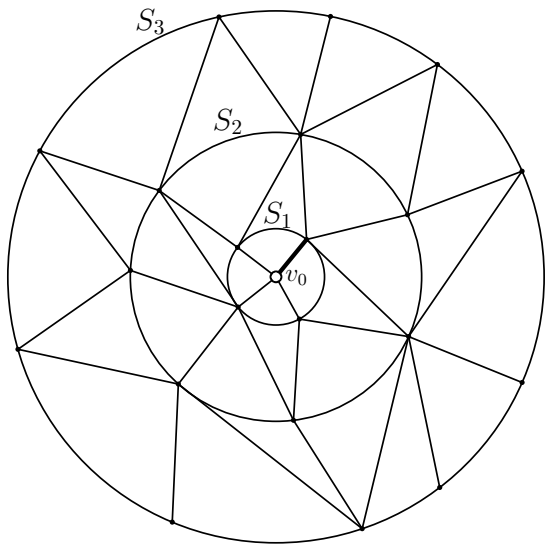}%
\caption{Causal triangulation with 3 slices. The root vertex is denoted by $v_0$ and the bold edge identifies the root edge.}%
\label{fig:ct}%
\end{figure}

The presence of a root edge allows us to unambiguously label the vertices of such triangulation. This can be done as follows. For each triangulation $T$ and for each $i\geq 1$ 
let us enumerate the vertices of the set $S_i,$ i.e., all the vertices at distance $i$ from the root as follows.


Let $v_0$ denote the root vertex and $v_{1,1}$ the endpoint of the root edge on $S_1$, and let us denote all the other vertices on $S_1$, taken in clockwise order starting from $v_{1,1}$, by $v_{1,2}$ up to $v_{1,|S_1|}$. Here $|S_i| \equiv |V(S_i)| = |E(S_i)|$, where for any set $A$ we denote by $|A|$ its cardinality.

Given $v_{i,1}, \ldots, v_{i,|S_i|}$ , 
take the endpoint of the leftmost edge connecting 
$v_{i,1}$ to $S_{i+1}$, this will be denoted by $v_{i+1,1}$,
 and proceed as above for all the other vertices on $S_{i+1} $.

Let us denote by $\mathcal{T}_N$, $N \in \mathbb{N}$, the set of causal triangulations with $N$ slices and by $\mathcal{T}_{N,\textbf{k}}$, $\mathbf{k} = (k_1,\dots,k_N)$, the set of causal triangulations with a fixed number of vertices on each slice, that is
\begin{equation}
	\mathcal{T}_{N,\mathbf{k}} = \{ T \in \mathcal{T}_N : |S_1| = k_1,\dots, |S_N| = k_N \}.
\end{equation}
According to this definition, the set $\mathcal{T}_N$ can be decomposed as follows
\begin{equation}
	\mathcal{T}_N = \bigcup_{k_1 = 1}^\infty \cdots \bigcup_{k_N = 1}^\infty \mathcal{T}_{N,\mathbf{k}}.
\label{eq:T_decom}
\end{equation}
Note that $T_{N,\mathbf{k}}$ is a finite set, in particular we have
\begin{equation}\label{eq:Tnk_card}
	|\mathcal{T}_{N,\mathbf{k}}| = \prod_{i=1}^{N-1} \binom{k_{i+1}+k_i-1}{k_i-1}.
\end{equation}

We shall often use also the notation $\mathcal{T}_{N,l}$ to denote the set of causal triangulations with fixed number $l$ of vertices on the last slice, that is 
\begin{equation}
	\mathcal{T}_{N,l} = \{ T \in \mathcal{T}_N : |S_N| = l \}.
\end{equation}

Given a finite triangulation $T \in \mathcal{T}_N$, let us denote the set of vertices of $T$ by $V(T)=\{v_0, v_{i,j}, 1\leq i\leq N, 1\leq j \leq |S_i|\}$, 
and the set of edges by $E(T)$, which is a subset of  $\{(u,v): u,v \in V(T)\}$. We shall also denote $F(T)$ the set of triangles in $T$.
It is easy to see that, given a finite triangulation $T \in \mathcal{T}_N$,  the number of vertices, edges and triangles in $T$ 
satisfy the following relations
\begin{equation}\label{vs}
	|V(T)|  = 1 + \sum_{i=1}^N |S_i|,
\end{equation}
\begin{equation}\label{tri}
	|F(T)|  = 2 \sum_{i=1}^{N-1} |S_i| + |S_N|,
\end{equation}
\begin{equation}\label{edges}
	|E(T)|  = 3 \sum_{i=1}^{N-1} |S_i| + 2 |S_N|. 
\end{equation}


In the following we decorate each finite triangulation with a spin configuration, defined as follows.

\begin{definition}\label{DSG}
A \textit{spin configuration} $\sigma$ on a graph $T$ with a set of vertices $V(T)$
 is an assignment of values $+1$ (\textit{spin up}) or $-1$ (\textit{spin down}) to each vertex, i.e.
\begin{equation}\label{sp}
        \sigma \: : \: V(T) \to \{+1,-1\}^{V(T)} \equiv \Omega(T).
\end{equation}
\end{definition}


Let $\Lambda_N$ denote the set of  finite triangulations with $N$ slices, together with spin configurations on them, 
i.e., 
 \begin{equation}\label{lsp}
        \Lambda_N = \left\{(T,\sigma(T)): T \in \mathcal{T}_N, \sigma(T) \in \Omega(T) \right\}.
\end{equation}

We call a {\it spin-graph of height $N$} an element of space $\Lambda_N$. Notice that any element of $\mathcal{T}_N$ is a finite graph, therefore $\Lambda_N$
is a set of finite spin-graphs.


\subsection{Gibbs family  on spin-graphs}

We shall define a probability measure on space 
$\Lambda_N$ and will study a random 
spin-graph sampled with respect to this measure. 
This is generally called an \textit{annealed coupling}.
The fundamental difference with  the \textit{quenched}  case
studied lately in \cite{KY} is that here 
the (Gibbs) measure is defined on a space of graphs with spins, unlike in a 
quenched case where a graph is first  sampled according to some measure on the set of graphs only, and then a Gibbs measure is defined on 
the configurations of spins on the sampled graph. 

One can find  in Malyshev \cite{Malyshev01} a general outline of the theory of Gibbs
families on spin graphs. We follow his approach and develop it here for a
class of planar triangulations.

The energy or Hamiltonian of the spin-graph $(T, \sigma)=(T,\sigma(T))$, where
$\sigma(T)=(\sigma_v)_{v \in V(T)} \in \{+1,-1\} ^{V(T)}$, is defined as
\begin{equation}\label{H}
        H(T, \sigma)= -\sum_{(u,v)\in E(T)}\sigma_u \sigma_v.
\end{equation}
Here $\sigma_u \sigma_v$ is a {\it potential} of an edge $(u,v)$, as in the well-known Ising model.

Define the partition function
\begin{equation}\label{Z}
     Z_N(\beta,\mu) = \sum_{ (T, \sigma) \in \Lambda_N} e^{-\beta H(T, \sigma)-\mu |F(T)|},
\end{equation}
where $\beta\geq 0$ is the inverse temperature, and $\mu$ is the 
\textit{cosmological constant}. 
Whenever this function is finite one can define  the following measure 
on $\Lambda_N$.

\begin{definition}\label{D1}
A {\it Gibbs distribution}  on the space of finite spin-graphs $\Lambda_N$
is a probability measure defined by
\begin{equation}\label{Gf}
     p_{N,\beta, \mu} (T, \sigma)   = \frac{e^{-\beta H(T, \sigma)-\mu |F(T)|}}{Z_N(\beta,\mu) }, \qquad (T,\sigma) \in \Lambda_N.
\end{equation}
\end{definition}

%

\subsection{Gibbs distributions on spin-graphs with fixed boundary conditions}

Consider now a graph $T\in \mathcal{T}_N$. It is natural to call 
the vertices of the outer slice $S_N$ of $T$ the {\it graph boundary} of the graph $T$. In the following we introduce  spin-graphs with a given
 {\it boundary condition}.

Given a triangulation $T \in \mathcal{T}_{N,l}$, a spin configuration on $T$ with \textit{boundary conditions} $\tilde{\sigma} \in \{+1,-1\}^l$ is an element of the set
\begin{equation}
	\Omega^{\tilde{\sigma}}(T) = \{ \sigma \in \Omega(T) : \sigma_v =  \tilde{\sigma}_v, v \in V(S_N)\}
\label{eq:Omega_bc}
\end{equation}
and a spin-graph $(T,\sigma)$ of height $N$ with $(l,\tilde{\sigma})$-boundary conditions  is an element of
\begin{equation}
	\Lambda_{N,l}^{\tilde{\sigma}} = \left\{(T,\sigma(T)): T \in \mathcal{T}_{N,l}, \ \ \sigma(T) \in \Omega^{\tilde{\sigma}}(T) \right\}.
\label{eq:lsp_bc}
\end{equation}

Similarly to (\ref{Gf}) one can define probability measures on the set $\Lambda_{N,l}^{\tilde{\sigma}}$.

\begin{definition}\label{Dbc}
A {\it Gibbs distribution}  on the space of finite spin-graphs $\Lambda_{N,l}^{\tilde{\sigma}}$ is a probability measure defined by
\begin{equation}\label{Gf_bc}
     p_{N,l,\beta, \mu}^{\tilde{\sigma}} (T, \sigma)   = \frac{e^{-\beta H(T, \sigma)-\mu |F(T)|}}{Z_{N,l}^{\tilde{\sigma}}(\beta,\mu) }, \ \ (T, \sigma) \in \Lambda_{N,l}^{\tilde{\sigma}},
\end{equation}
where 
\begin{equation}\label{Z_bc}
     Z_{N,l}^{\tilde{\sigma}}(\beta,\mu) = \sum_{ (T, \sigma) \in \Lambda_{N,l}^{\tilde{\sigma}}} e^{-\beta H(T, \sigma)-\mu |F(T)|}.
\end{equation}
\end{definition}

When $\tilde{\sigma}_v = +1$, for all $v \in V(S_N)$, the set of spin-graphs with fixed boundary, the measure on it and the partition function will be denoted by $\Lambda_{N,l}^+$, $p_{N,l,\beta, \mu}^+$ and $Z_{N,l}^+$, respectively. Also, if 
$\tilde{\sigma}_v = -1$, for all $v \in V(S_N)$, we use notations
 $\Lambda_{N,l}^-$, $p_{N,l,\beta, \mu}^-$ and $Z_{N,l}^-$, respectively.

\section{Main results}
\label{sec:results}
\subsection{Finite triangulations.}

First, we note that, since the set of causal finite 
triangulations $\mathcal{T}_N$ is countable, the sum in (\ref{Z}) might be divergent. Hence, the partition function (\ref{Z}) (or that with boundary conditions  (\ref{Z_bc})) can be infinite. Since the partition function is decreasing with $\mu$ (when other parameters are fixed), we can define for any fixed $N$ and  $\beta$ (and
${\tilde{\sigma}}$) the critical values $\mu_{cr}^N(\beta)$ such that
\begin{equation}
	Z_N(\beta,\mu)<\infty, \mbox{ if } \mu>\mu_{cr}^N(\beta),
\end{equation}
and
\begin{equation}
	Z_N(\beta,\mu)=\infty, \mbox{ if }  \mu<\mu_{cr}^N(\beta).
\end{equation}

In the following theorem we provide bounds on the region of the $(\beta,\mu)$-plane where the partition function is finite for all $N \in \mathbb{N}$.

\begin{theorem} \label{thm:cr_bounds}
The partition function $Z_N(\beta,\mu)$ is finite for all $N \in \mathbb{N}$  in the region of the $(\beta,\mu)$-plane defined by
\begin{equation}
	\Delta_{f} = \{ (\beta,\mu) \in \mathbb{R}^2 : \beta \geq 0, \mu > \beta + \log(1+ 2 \cosh\beta)\}.
\label{eq:sub_reg}
\end{equation}
Moreover, if $Z_N(\beta,\mu)$ is finite for all $N$, then we necessarily have
\begin{equation}
	\mu > \max\{\log(1+ \cosh \beta + \cosh(2\beta)), \beta + \log(1+e^\beta)\}.
\label{eq:low_bound}
\end{equation}
\end{theorem}

The theorem is proved in Sec. \ref{sec:proof1}. In Fig. \ref{fig:Z_bounds} the above bounds are shown. We note that the bounds coincide at $\beta = 0$ and $\beta \to \infty$. In particular, as a direct consequence of the above theorem, we have the following corollary, which reproduces the known result \cite{MYZ01} for causal triangulations without spins.

\begin{center}
\begin{figure}%
\includegraphics[scale=1.05]{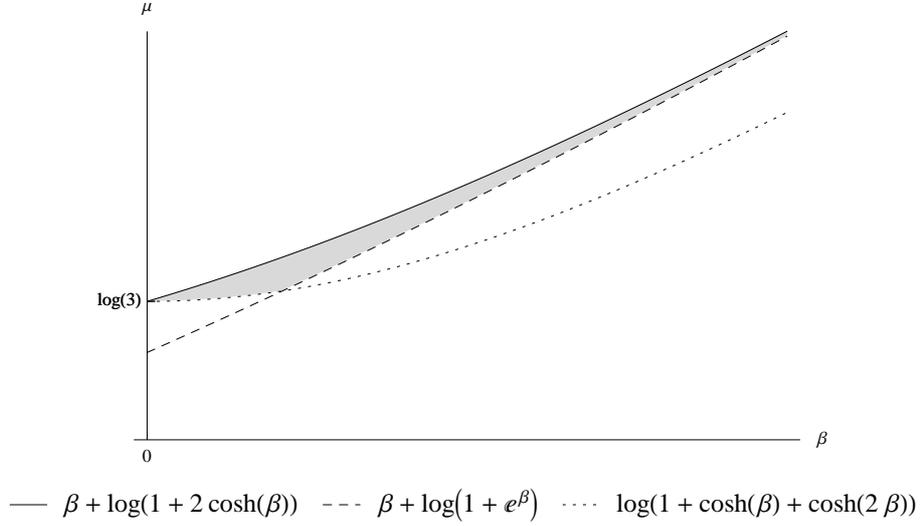}
\caption{Bounds defined in Thm. \ref{thm:cr_bounds}. The critical line is located within the shaded region.}%
\label{fig:Z_bounds}%
\end{figure}
\end{center}

\begin{corollary}
At $\beta = 0$ the partition function is finite for all $N$, if and only if $\mu > \log 3$.
\end{corollary}

Note that the difference between the critical value given in the above Corollary and the value found in \cite{MYZ01} is simply due to the summation over spins, which is trivial for $\beta=0$.


Define now the \textit{critical line}
\begin{equation}
	\mu_{cr}(\beta) = \sup_{N \in \mathbb{N}} \mu_{cr}^N(\beta).
\label{mucr}
\end{equation}
From Thm. \ref{thm:cr_bounds} we have that  
\begin{equation}
		\max\{\log(1+ \cosh \beta + \cosh(2\beta)), \beta + \log(1+e^\beta)\} \leq \mu_{cr}(\beta) \leq \beta + \log(1+ 2 \cosh\beta).
\end{equation}
It follows from the definition of $\mu_{cr}(\beta)$ that
$Z_N(\beta,\mu)$ is finite for all $N \in \mathbb{N}$ if $\mu>\mu_{cr}(\beta)$, but if $\mu<\mu_{cr}(\beta)$ then at least for some $N$ the partition function is infinite. 

Furthermore, Thm. \ref{thm:cr_bounds} gives us the following information on the critical parameters for $\beta>0$.

\begin{corollary}
The critical line $\mu_{cr}(\beta)$ is continuous at $\beta = 0$, in particular
\begin{equation}
	\lim_{\beta \to 0} \mu_{cr}(\beta) = \log 3.
\end{equation}
\end{corollary}

\begin{corollary}
When $\beta$ is large the critical line has the asymptotics
\begin{equation}
	\mu_{cr}(\beta) = \beta + \log(1+e^\beta) + O(e^{-2\beta}) \quad \text{as } \beta \to \infty.
\end{equation}
\end{corollary}

A quantity that would provide us information on the influence of the boundary conditions on the magnetization properties of the spin system is given by the \textit{mean magnetization} of the central spin $\sigma_0=\sigma_{v_0}$ (i.e. the spin attached to the root vertex $v_0$). In view of  Definition \ref{Dbc}, this is given by
\begin{equation}
	\begin{split}
		\langle \sigma_0 \rangle_{N,l,\beta,\mu}^{\tilde{\sigma}} & := 
\sum_{ \substack{(T, \sigma) \in \Lambda_{N,l}^{\tilde{\sigma}}}} \sigma_0 \ p_{N,l,\beta, \mu}^{\tilde{\sigma}} (T, \sigma) \\
		& =\sum_{ \substack{(T, \sigma) \in \Lambda_{N,l}^{\tilde{\sigma}}\\ \sigma_0 = +1}} p_{N,l,\beta, \mu}^{\tilde{\sigma}} (T, \sigma) - \sum_{ \substack{(T, \sigma) \in \Lambda_{N,l}^{\tilde{\sigma}}\\ \sigma_0 = -1}} p_{N,l,\beta, \mu}^{\tilde{\sigma}} (T, \sigma).
	\end{split}
\label{eq:s0p_expect}
\end{equation}

Notice that due to the symmetry in the model without fixed boundary 
we have for all $N$ and $l$
\begin{equation}
\label{UB}
\langle \sigma_0 \rangle_{N,l,\beta,\mu}= 0.
\end{equation}

The following result shows that (depending on the parameters) in the limit $N \to \infty$ the magnetization of the central spin is unaffected by the remote boundary conditions.

\begin{theorem}\label{thm:mag}
For $\beta$ small enough and $\mu > 3/2 \log (\cosh \beta) + 3 \log 2 $, the mean magnetization of the central spin of spin-graphs with $(l,-)$- as well as with $(l,+)$-boundary conditions converges to 0 as $N$ goes to infinity, that is
\begin{equation}
	\lim_{N \to \infty}\langle \sigma_0 \rangle_{N,l, \beta,\mu}^+ = 0=\lim_{N \to \infty}\langle \sigma_0 \rangle_{N,l, \beta,\mu}^- .
\end{equation}
\end{theorem}

Observe however, that unlike in (\ref{UB}) here for any finite $N$
we have the following. 

\begin{theorem} \label{TC}
For any $(\beta,\mu )\in \Delta_{f} $ defined in Theorem \ref{thm:cr_bounds},   and any finite $N$
one has
\begin{equation}
	\langle \sigma_0 \rangle_{N,l, \beta,\mu}^- < 0 < \ \ \langle \sigma_0 \rangle_{N,l, \beta,\mu}^+.
\end{equation}
\end{theorem}
Notice that the statement of Theorem \ref{TC} is certainly expected but 
still non-trivial since the set of triangulations $\mathcal{T}_{N,l}$ is countable even for finite $N$.

\subsection{Towards constructing a measure on infinite triangulations.}
Let ${\mathcal T}_{\infty}$ be a set of {\it infinite} rooted
triangulations with a countable number of slices but with  
finite numbers of vertices on each slice $S_i$, $i\geq 1$. 
One aims  to construct a measure 
on the space of {\it infinite} spin-graphs 
\begin{equation}
	\Lambda_{\infty}=\{(T, \sigma) : T\in {\mathcal T}_{\infty}, \ \sigma \in \{-1,+1\}^{V(T)}\}.
\end{equation}

For any $T\in {\mathcal T}_{\infty}$ and  $N \in \mathbb{N}$ 
let $T|_{{\mathcal T}_N}$ 
denote a subgraph of $T$ on the vertices at distance at most $N$ from 
the root. 
For any $N$ and $T\in {\mathcal T}_N$ 
define a cylinder set 
\begin{equation}
	{\mathcal C}_{\infty}(T):=\{T'\in {\mathcal T}_{\infty}: 
T'|_{ {\mathcal T}_N}=T \}.
\end{equation}
To define a measure on the $\sigma$-algebra generated by the cylinder sets 
a usual way (for lattices, for example) 
is to construct first a Gibbs family of conditional distributions. 

Below we define Gibbs measure
on finite spin-graphs with boundary conditions.  
We will show (Lemma \ref{Mp})
that the family of conditional distributions constructed from the introduced above Gibbs family is consistent. This  gives a ground for the 
Dobrushin-Lanford-Ruelle construction of  Gibbs measure on ${\mathcal T}_{\infty}$ (see \cite{Malyshev01}).

\subsubsection{Gibbs family of conditional  distributions.}
\label{CD}

Let us  introduce a more general space of triangulations. 
For any $0\leq K\leq N$ and $k, n\geq 1$ let $\mathcal{T}_{K,N}(k,n)$
denote the set of all rooted triangulations defined as above, but
 whose vertices belong to the slices $S_K, \ldots, S_N$, where $|S_K|=k$, 
$S_N=n$. 
Let also $\mathcal{T}_{K,N}=\bigcup_{k,n}\mathcal{T}_{K,N}(k,n)$.

Given a graph $T \in \mathcal{T}_N$ define for any $0\leq K<N$ a subgraph 
of $T$, on the set of vertices which consists of the root and of all the vertices on the first $K$ slices; denote this subgraph by 
$T|_{\mathcal{T}_K}$. In other words, $T|_{\mathcal{T}_K}$ is the subgraph of $T$ spanned by the vertices at distance at most $K$ from  the root. 

For any graph $T$ and its subgraph $G$ define $T\setminus G$ to be a subgraph of $T$ on the vertices $V(T)\setminus V(G)$. Observe that with this definition we have
\begin{equation}
	V(G)\cup V(T \setminus G) =V(T),
\end{equation}
however
\begin{equation}
	E(G)\cup E(T \setminus G)  \subset E(T),
\end{equation}
since in the set on the left we do not have the edges which connect
 vertices of $G$ to the vertices of $T\setminus G$. 
Therefore, given a graph (a rooted triangulation) 
we can define uniquely  (with respect to the root and to 
the order of vertices on the slice) a subgraph, as well as  the 
complement subgraph. However, there is no
 a one-to-one correspondence here
 since we can join two rooted subgraphs in different ways (even with 
preserved order on the slice). This leads to a following definition
of a union of two graphs. 

\begin{definition}\label{DU}
For any $0\leq K< N$ a union of two graphs 
$T_K \in  {\mathcal T}_K$ and  $ {\widetilde T}\in 
{\mathcal T}_{K+1,N}
$ is a {\it subset of graphs} in $ {\mathcal T}_N$:
\begin{equation}
	T_K \cup {\widetilde T} 
:=\{ T\in {\mathcal T}_N: T|_{{\mathcal T}_K}=T_K, \ T\setminus T_K= {\widetilde T}\}.
\end{equation}
\end{definition}

Definition \ref{DU} gives  us a natural representation of the set
$\Lambda_{N+1,k }^{\widetilde\sigma} $:
\begin{equation}\label{BU}
\Lambda_{N+1,k }^{\widetilde\sigma}   = \bigcup_{(T, \sigma)\in \Lambda_{N}}\{
(T',\sigma'): T'\in T\cup S_{N+1}, |S_{N+1}|=k, 
\sigma'=(\sigma, {\widetilde\sigma})\},
\end{equation}
where we write $\sigma'=(\sigma, {\widetilde\sigma})$ if $\sigma'(v)=\sigma(v), v \in V(T),$ and $
\sigma'(v)={\widetilde\sigma}(v), v \in V(S_{N+1})$.
The Gibbs distribution  (\ref{Gf_bc}) induces 
 the following 
 probability measure on $\Lambda_{N}$.

\begin{definition}\label{DBC} 
For any  $N\geq 0$ and  $k\geq 1$
define a slice 
\begin{equation}
	S_{N+1}=(v_{N+1, 1}, \ldots , v_{N+1, |k|}).
\end{equation}
Then for any ${\widetilde \sigma}\in \{-1, +1\}^k$, the Gibbs distribution  (\ref{Gf})
defines a conditional probability on $\Lambda_{N}$:
\begin{equation}\label{GfBC1}
p_{N,\beta, \mu}
 \left\{ (T, \sigma)|(S_{N+1}, {\widetilde\sigma})  \right\}
= \frac{
\sum_{T' \in T\cup S_{N+1}}  e^{-\beta
H(T', (\sigma, {\widetilde\sigma}))-\mu F(T)}}{
Z_{N+1,k}^{\widetilde\sigma} (\beta, \mu)
 }, \ \ \  (T, \sigma)\in \Lambda_{N},
\end{equation}
which is called a  conditional Gibbs distribution 
 with  the  boundary condition $ {\widetilde\sigma}$.
\end{definition}
In view of the last definition, the probability measure defined in (\ref{Gf}) is also  called a Gibbs distribution with {\it free boundary conditions}. 

Observe that the probability in Definition \ref{DBC}  is defined on the entire $\Lambda_{N}$,
unlike the one defined in  (\ref{Gf_bc}), which is on  $\Lambda_{N,l}^{\widetilde\sigma} \subset \Lambda_{N}.$

We shall consider now a more general class of 
conditional  Gibbs distributions.
Since for any  $0\leq K< N$ any graph $T\in {\mathcal T}_N$
has a subgraph $T|_{{\mathcal T}_K}\in {\mathcal T}_K$, 
the Gibbs distribution 
(\ref{Gf})   on $\Lambda_N$ induces as well the conditional 
probability on 
$\Lambda_K$ as we define below.

\begin{definition}\label{Dc}
For any $0\leq K< N$
and  $ ({\widetilde T}, {\widetilde \sigma})\in 
\Lambda_{K+1,N}
$ define
\begin{equation}\label{J7}
     p_{N,\beta, \mu} 
\left\{ 
(T_K, \sigma_K) \mid
 ({\widetilde T},  {\widetilde \sigma})
\right\}
  := \frac{ \sum_{ T\in T_K \cup {\widetilde T} }\ 
p_{N,\beta, \mu} (T, (\sigma_K,  {\widetilde \sigma}))
}{
\sum_{  (T'_K, \sigma'_K ) \in {\Lambda}_K   }
\sum_{
T\in  T'_K \cup {\widetilde T} \ 
}
 p_{N,\beta, \mu} (T, (\sigma'_K,  {\widetilde \sigma}))
}, 
\end{equation}
for all $(T_K, \sigma_K)\in \Lambda_K$, which is a conditional distribution on $\Lambda_K$.
\end{definition}

In the following lemma we derive a Markov property of the last conditional distribution, by proving a simple relation between the conditional probability (\ref{J7}) and the Gibbs distribution given in Definition \ref{DBC}.
\begin{lemma}\label{Mp}
For any $0\leq K< N$
and  $ ({\widetilde T}, {\widetilde \sigma})\in 
\Lambda_{K+1,N}$
one has the following equalities
\begin{equation}\label{p1}
	\begin{split}
		p_{N,\beta, \mu} \left\{ (T_K, \sigma_K) \mid  ({\widetilde T},  {\widetilde \sigma}) \right\} & = \frac{ \sum_{ T\in T_K \cup {\widetilde S}_{K+1} } 
e^{-\beta H(T, (\sigma_K,  {\widetilde \sigma}_{K+1}))-\mu F(T)} }{Z_{ K+1,|{\widetilde S}_{K+1}|}^{\widetilde\sigma}(\beta, \mu)} \\
		& = p_{K,\beta, \mu} \left\{ \left(T_K, \sigma_K \right) \mid   \left( |{\widetilde S}_{K+1}|, {\widetilde \sigma}_{{K+1}} \right) \right\},
	\end{split}
\end{equation}
where $\left( {\widetilde S}_{K+1}, {\widetilde \sigma}_{{K+1}} \right)$ denotes the $(K+1)$-st slice with spins of the given spin-graph $({\widetilde T}, {\widetilde \sigma})$.
\end{lemma}

The last equality follows  simply by  Definition \ref{DBC}, it underlines  
 that the conditional probability in  (\ref{p1})
depends only on the spins on the vertices of ${\widetilde T}$ which are 
connected to $T_K$, i.e., only those which interact with the boundary. This reflects the Markov property  of the conditional probabilities on the left. The first equality shows that the conditional distribution is again in the form of (\ref{Gf}), which confirms the Gibbs property of the conditional distribution. Observe, that the formula in  (\ref{p1}) does not depend on $N$ (as long as $K<N$).

\begin{remark}
\label{R1} 
Theorem \ref{thm:mag} and Theorem \ref{TC} hold as well if the Gibbs distribution with fixed boundary conditions is replaced by the conditional distribution (\ref{p1}).
\end{remark}

\section{ Proof of Theorem \ref{thm:cr_bounds}.}
\label{sec:proof1}

We shall study here the partition function defined in (\ref{Z}). 
Let us rewrite it 
using the space $\mathcal{T}_{N,\textbf{k}}$, $\mathbf{k} = (k_1,\dots,k_N)$, as
\begin{equation}
	\begin{split}
		Z_N(\beta,\mu) & = \sum_{l\geq 1} \sum_{ (T, \sigma) \in \Lambda_{N,l} } e^{-\beta H(T, \sigma)-\mu |F(T)|} \\
		& = \sum_{l\geq 1} \sum_{\textbf{k}:k_N=l} \sum_{ T \in  \mathcal{T}_{N,\textbf{k}}} e^{-\mu |F(T)|} \sum_{ \sigma \in \Omega(T)} e^{-\beta H(T, \sigma)},
	\end{split}
\label{Z1}
\end{equation}
where we decomposed the sum in eq. (\ref{Z}) according to the number $l$ of vertices on $S_{N}$.
\subsection{Upper bound}

First we observe that the Hamiltonian (\ref{H})  for any $T\in \mathcal{T}_{N,\textbf{k}}$ 
can be written by splitting the interaction between spins on different slices and spins on the same slice, that is
\begin{equation}
	\begin{split}
	H(T,\sigma) & = - \sigma_0 \sum_{i \in V(S_1)} \sigma_i - \sum_{\substack{i \in V(S_1),  j \in V(S_2) \\ (i,j) \in E(T)}} \sigma_i \sigma_j - \dots - \sum_{\substack{i \in V(S_{N-1}),  j \in V(S_N) \\ (i,j) \in E(T)}} \sigma_i \sigma_j \\
	& - \sum_{i = 1}^{k_1} \sigma_{1,i} \sigma_{1,i+1} -  \dots - \sum_{i = 1}^{k_N} \sigma_{N,i} \sigma_{N,i+1},
	\end{split}
\label{eq:Zup_Hsplit}
\end{equation}
where $\sigma_{i,k_i+1} = \sigma_{i,1}$, $i = 1,\dots,N$. Therefore, considering that for any $u,v \in V(T)$, $\sigma_u \sigma_v \in \{+1,-1\}$ and that the number of edges connecting the slice $S_i$ and $S_{i+1}$ is $k_i+k_{i+1}$, from eq. (\ref{eq:Zup_Hsplit}) we obtain for any $T\in \mathcal{T}_{N,\textbf{k}} $ with $k_N=l\geq 1$,
\begin{equation}
	H(T,\sigma) \geq  - 2 \sum_{i = 1}^{N-1} k_i -l+ \sum_{i=1}^N H_1(\sigma^i,k_i),
\label{I1}
\end{equation}
where $\sigma^i=(\sigma_v)_{v\in S_i}=(\sigma_{i,1}, \ldots, \sigma_{i,k_i})$, and 
\begin{equation}
	H_1(\sigma^i,k_i) = - \sum_{j = 1}^{k_i} \sigma_{i,j} \sigma_{i,j+1}
\end{equation}
is the Hamiltonian of a 1-dimensional Ising model with $k_i$ spins and periodic boundary conditions. From eq. (\ref{Z1}), using the above inequality (\ref{I1}) we get
\begin{equation}
	\begin{split}
		Z_{N}(\beta,\mu)  & \leq \sum_{l \geq 1} e^{\beta l} \sum_{\textbf{k}:k_N=l} \sum_{ T \in  \mathcal{T}_{N,\textbf{k}}} e^{2\beta \sum_{i=1}^{N-1} k_i- \mu |F(T)|} \sum_{\sigma \in \Omega_N(T)} \prod_{j = 1}^N e^{ - \beta H_1(\sigma^i,k_i)} \\
		& = \sum_{l \geq 1}  e^{(\beta-\mu)  l} \sum_{\textbf{k} :k_N=l} e^{2(\beta-\mu) \sum_{i=1}^{N-1} k_i}\prod_{i=1}^{N-1} \binom{k_{i+1}+k_i-1}{k_i-1} \\
		& \times \prod_{j = 1}^N  \left( \sum_{\sigma \in \{-1, +1\}^{k_j} } e^{ - \beta H_1(\sigma,k_j)}\right)
		\end{split}
\label{I2}
\end{equation}
where the last equality is due to  equations (\ref{tri}) and (\ref{eq:Tnk_card}).
Using then the following well-known result for the partition function of a 1-dimensional Ising model (see e.g. \cite{Baxter}) 
\begin{equation}
	\sum_{\sigma \in \{-1, +1\}^{k_j}} e^{ - \beta H_1(\sigma,k_j)} = (2 \sinh\beta)^{k_j} + (2 \cosh\beta)^{k_j} \leq 2 (2 \cosh\beta)^{k_j},
\end{equation}
we derive from  (\ref{I2}) 
\begin{equation}
	\begin{split}
		Z_{N}(\beta,\mu)  \leq & 2^{N+1} \sum_{l \geq 1} (2 e^{\beta-\mu} \cosh\beta)^l \\
		& \times \sum_{k_1, \dots, k_{N-1}} \prod_{i=1}^{N-1} \binom{k_{i+1}+k_i-1}{k_i-1} (2 e^{2 \beta -2 \mu}  \cosh\beta)^{k_i}.
	\end{split}
\label{I3}
\end{equation}

We shall use the following lemma (which is proved in \cite{MYZ01}, but we 
provide the proof   in Appendix \ref{apx:series} with some additional details that we use here).

\begin{lemma}(\cite{MYZ01}) \label{LA}
Define for $x\geq 0$ and $n,l \geq 1$ 
\begin{equation}
	W_{n+1,l}(x) = \sum_{k_1 = 1}^\infty \cdots \sum_{k_n = 1}^\infty \prod_{i=1}^{n} \binom{k_{i+1}+k_i-1}{k_i-1} x^{k_i},
\label{e1}
\end{equation}
with $k_{n+1}=l$. 
We have three cases:
\begin{enumerate}
	\item if $0<x < 1/4$ then $W_{n,l}(x)$ is finite for all $n,l \in \mathbb{N}$,
and in particular,
\begin{equation}
	W_{n+1,l}(x) \sim (1-4x) \left( \frac{2}{1+\sqrt{1-4x}} \right)^{l+3}  \left( \frac{2\sqrt{x}}{1+\sqrt{1-4x}} \right)^{2n}, \quad n \to \infty;
\label{eq:W_asym}
\end{equation}
	\item $W_{n+1,l}(1/4)$ is finite for all $n,l \in \mathbb{N}$,
and in particular,
		\begin{equation}
			\lim_{n \to \infty} W_{n+1,l}(1/4)  \sim 2^{l+1}, \quad n \to \infty;
		\end{equation}
	\item if
$\frac{1}{4} < x < \left[ 2 \cos\left( \frac{\pi}{n+2} \right)\right]^{-2} $, 
then $W_{n,l}$ is finite for all $l \in \mathbb{N}$.
\end{enumerate}
\end{lemma}

Notice that in the notations of \cite{MYZ01} function $W_{n+1,l}(x)$ is 
\begin{equation}
	e^{ \mu l}\sum_{k\geq 1}e^{-\mu k}Z_{[1,n+1]}(k,l)
\end{equation}
with $\mu=-2\log{x}.$

Observe  that, in the third case the length of the interval of the values $x$ where $W_{n,l}(x)$ is finite shrinks to 0 for $n \to \infty$. Therefore $W_{n,l}(x)$
is finite for a fixed $x$ and all $n$ if and only if $x\leq 1/4$.

\begin{corollary}\label{CA} The function
\begin{equation}
	W_{n+1}(x,y) = \sum_{l\geq 1}y^lW_{n+1,l}(x),
\end{equation} 
defined for positive $x$ and $y$, is finite for all $n$ if and only if 
\begin{equation}
	x\in (0, 1/4] \quad \text{ and } \quad y^2-y+x<0.
\label{eq:}
\end{equation}
\end{corollary}

Rewrite now inequality (\ref{I3}) as
\begin{equation}
	Z_{N}(\beta,\mu)  \leq 2^{N+1} W_{N}(x,y),
\label{I4}
\end{equation}  
where $y=2 e^{\beta-\mu} \cosh\beta$ and 
$x=2 e^{2 \beta -2 \mu}  \cosh\beta=e^{\beta-\mu}y$. Then it is easy to check that the conditions of Corollary \ref{CA} are satisfied if 
\[\mu > \beta + \log(1+ 2 \cosh\beta). \]
This proves the first part of Theorem \ref{thm:cr_bounds}.

\subsection{Lower bound}
\subsubsection{High-temperature region}
In the following, we rewrite the partition function in eq. (\ref{Z}) applying a classical high-temperature expansion argument to our case. 

First, we note that
\begin{equation}
e^{\beta \sigma_i \sigma_j}= (1 + \sigma_i \sigma_j \tanh \beta) \cosh \beta
.
\end{equation}
Hence, using  (\ref{H}), for any $\sigma \in \Omega (T)$ and $T \in \mathcal{T}_{N,l}$ we get 
\begin{equation}\label{eq:hte}
\begin{split}
	e^{-\beta H(T,\sigma)} & = \prod_{(i,j) \in E(T)} e^{\beta \sigma_i \sigma_j} \\
	& =  (\cosh \beta)^{|E(T)|} \prod_{(i,j) \in E(T)} (1 + \sigma_i \sigma_j \tanh \beta) \\
	& =  (\cosh \beta)^{|E(T)|}  \sum_{E \subseteq E(T)} \prod_{(i,j) \in E} \sigma_i \sigma_j \tanh\beta,
\end{split}
\end{equation}
where the sum runs over all subsets $E$ of $E(T)$, 
including the empty set, which gives
contribution 1. For any vertex $i\in V(T)$ and a subset of edges
$E \subseteq E(T)$ let us define the incidence number
\begin{equation}
	I(i,E) = |\{ j \in V(T) : (i,j) \in E \}|.
\end{equation}
With this notation we derive from  (\ref{eq:hte}) 
\begin{equation}\label{e2}
	e^{-\beta H(T,\sigma)} 	 =  (\cosh \beta)^{|E(T)|}  
\sum_{E \subseteq E(T)} (\tanh \beta)^{|E|} \prod_{i \in V(T)} 
\sigma_i^{I(i,E)}.
\end{equation}

Next, we note that
\begin{equation}
	\sum_{\sigma_i = \pm 1} \sigma_i^{I(i,E)} =
	\begin{cases}
	 2, & \text{if $I(i,E)$ is even,} \\
	 0, & \text{if $I(i,E)$ is odd.}
	\end{cases}
\end{equation}
Therefore, 
\begin{equation}
	\sum_{\sigma \in \Omega(T)} \sum_{E \subseteq E(T) } (\tanh \beta)^{|E|} \prod_{i \in V(T)} \sigma_i^{I(i,E)} = 2^{|V(T)|} \sum_{E \in \mathcal{E}(T)} (\tanh \beta)^{|E|},
\label{e3}
\end{equation}
where
\begin{equation}
	\mathcal{E}(T) = \{ E \subseteq E(T)  : I(i,E) \, \text{is even} \, \forall i \in  V(T) \}.
\label{setE}
\end{equation}
Combining now (\ref{e3}) and  (\ref{eq:hte}) we get
\begin{equation}\label{e4}
	\sum_{\sigma \in \Omega(T)} 	e^{-\beta H(T,\sigma)} 	 =  (\cosh \beta)^{|E(T)|}   2^{|V(T)|}
\sum_{E \in \mathcal{E}(T)} (\tanh \beta)^{|E|}.
\end{equation}
Substituting the last formula into (\ref{Z1}) 
allows us to rewrite the partition function as
\begin{equation}
     Z_N(\beta,\mu) = 
\sum_{l\geq 1} \sum_{ T \in  \mathcal{T}_{N,l}} 
e^{-\mu |F(T)|} (\cosh \beta)^{|E(T)|}   2^{|V(T)|}
\sum_{E \in \mathcal{E}(T)} (\tanh \beta)^{|E|}.
\label{HT}
\end{equation}
Applying relations  (\ref{vs}),  (\ref{tri}) and  (\ref{edges}), 
we derive from here
\begin{equation}
	\begin{split}
	Z_{N}(\beta,\mu) & = 2\sum_{l\geq 1} (2e^{-\mu} (\cosh \beta)^2)^l\sum_{\textbf{k}:k_N=l} \sum_{ T \in  \mathcal{T}_{N,\textbf{k}}} \left(\prod_{i=1}^{N-1} (2 e^{-2 \mu} (\cosh \beta)^3)^{k_i} \right) \\
	& \times \sum_{E \in \mathcal{E}(T)} (\tanh \beta)^{|E|},
	\end{split}
\label{eq:Z_hte}
\end{equation}
which yields 
\begin{equation}
	Z_{N}(\beta,\mu) \geq 2\sum_{l\geq 1} (2e^{-\mu} (\cosh \beta)^2)^l \sum_{T \in \mathcal{T}_{N,l}} \left(\prod_{i=1}^{N-1} (2 e^{-2 \mu} (\cosh \beta)^3)^{k_i} \right).
\label{eq:Z_loboundht}
\end{equation}
Using first (\ref{eq:Tnk_card}),  and then using the function defined in  
Corollary \ref{CA}, we derive from (\ref{eq:Z_loboundht})
\begin{equation}
	\begin{split}
		Z_{N}(\beta,\mu) & \geq 2\sum_{l\geq 1} (2e^{-\mu} (\cosh \beta)^2)^l \\
		& \times \sum_{k_1\geq 1, \ldots , k_{N-1}\geq 1} \prod_{i=1}^{N-1} \binom{k_{i+1}+k_i-1}{k_i-1}\left(2 e^{-2 \mu} (\cosh \beta)^3\right)^{k_i} \\ 
		& = 2 W_N(x,y),
	\end{split}
\end{equation}
where $y=2e^{-\mu} (\cosh \beta)^2$,
$x=2 e^{-2 \mu} (\cosh \beta)^3=e^{-\mu} \cosh \beta \ y$.
Therefore, if $Z_{N}(\beta,\mu)$ is finite for all $N$, it follows from Corollary
 \ref{CA} that
\begin{equation}\label{eq:mu_loboundht}
	\mu > \log(1+ \cosh \beta + \cosh(2\beta)).
\end{equation}

\subsubsection{Low-temperature region}
A lower bound for the partition function (\ref{Z}) is provided by the contributions of the two ground state configurations, that is all spins are "up'' or all spins are "down''. Therefore, we have
\begin{equation}
	\begin{split}
		Z_{N}(\beta,\mu) &  \geq 2 \sum_{l \geq 1} \sum_{T \in \mathcal{T}_{N,l}} e^{\beta |E(T)| - \mu |F(T)| } \\
		& = 2 \sum_{l \geq 1} e^{(2 \beta - \mu)l} 
\sum_{k_1\geq 1, \ldots , k_{N-1}\geq 1} 
\prod_{i=1}^{N-1} \binom{k_{i+1}+k_i-1}{k_i-1}
e^{(3\beta-2\mu)k_i}\\
& = 2 W_N\left(e^{3\beta-2\mu},  e^{2\beta-\mu}\right),
	\end{split}
\end{equation}
where the last function is again as defined in Corollary \ref{CA}.
Due to the result of Corollary \ref{CA} 
the partition function is finite for all $N$ only if
\begin{equation} \label{eq:mu_loboundlt}
	\mu > \beta + \log(1+e^\beta).
\end{equation}

Putting together the bounds (\ref{eq:mu_loboundht}) and (\ref{eq:mu_loboundlt}), the second part of Theorem \ref{thm:cr_bounds} follows. This finishes the proof of Theorem \ref{thm:cr_bounds}. $\Box$

\section{Behaviour of the partition function}

\begin{lemma}
The partition function $Z_{N}(\beta,\mu)$ is a continuous function of $\beta$ at $\beta = 0$, in particular 
\begin{equation}
	Z_{N}(0,\mu) = \lim_{\beta \searrow 0} Z_{N}(\beta,\mu) = 2 W_{N}(2 e^{-2\mu},2 e^{- \mu}).
\end{equation}
\end{lemma}
\begin{proof}
First we note that
\begin{equation}
	H(T,\sigma) \geq - |E(T)|.
\end{equation}
Therefore,
\begin{equation}
	\begin{split}
		Z_{N}(\beta,\mu) & \leq \sum_{l \geq 1} \sum_{T \in \mathcal{T}_{N,l}} 2^{|V(T)|} e^{\beta |E(T)| - \mu |F(T)|} \\
		& = 2 \sum_{l \geq 1} (2 e^{2 \beta - \mu})^l \sum_{\textbf{k}:k_N=l} \sum_{ T \in  \mathcal{T}_{N,\textbf{k}}} \prod_{i=1}^{N-1} (2 e^{3\beta-2\mu})^{k_i}.
	\end{split}
\label{eq:Z_upbound}
\end{equation}
Hence, from  (\ref{eq:Z_upbound}) and (\ref{eq:Z_loboundht}) we get
\begin{equation}
	2 W_{N}(2 e^{-2 \mu} (\cosh \beta)^3,2e^{-\mu} (\cosh \beta)^2) \leq Z_{N}(\beta,\mu) \leq 2 W_{N}(2 e^{3\beta-2\mu},2 e^{2 \beta - \mu}),
\end{equation}
for $\mu > \beta + \log(1+2e^\beta)$, and the result follows. \qed
\end{proof}

\section{Relation with previous results. Duality}

A model similar to the one studied in the present paper has been previously investigated in \cite{HSYZ} and \cite{Hern}. There, the authors consider the annealed coupling of causal triangulations of a torus with an Ising model, where the spins are attached on the faces of the triangulation, that is on the vertices of the dual graph of the triangulation (not on the vertices of the triangulation itself, as in our case):
\begin{equation}
	\sigma^{\Delta}=(\sigma_t, t\in F(T))\in \{-1,1\}^{|F(T)|},
\end{equation}
and  the Hamiltonian of the spin-graph $(T, \sigma^{\Delta})$, is given by
\begin{equation}\label{Htri}
        H^{\Delta}(T, \sigma^{\Delta})= -\sum_{t\sim t' }\sigma_t \sigma_{t' },
\end{equation}
where the sum runs over the adjacent triangles $t,t' \in F(T)$.

In particular, in \cite{HSYZ} the authors introduce a transfer-matrix formalism which is used to obtain some bounds on the critical line. 
Then, in \cite{Hern} the author uses  Fortuin-Kasteleyn 
formalism of random-cluster model 
of the partition function, which also adds some details to the analysis in \cite{HSYZ}, 
(although  it does not improve the bounds of \cite{HSYZ} for the critical parameters).

Note that, though our model and the one in \cite{HSYZ} are similar, a quantitative comparison between the two models is not straightforward and the results found in a model do not automatically follow from the other. In particular, this is due to the different geometries of the triangulated manifolds on which the spin system lie, the plane in our case and the torus in \cite{HSYZ} and \cite{Hern}.

On the other hand, the two models agree at least qualitatively on the behaviour of the critical line, as a comparison between the phase diagrams in the two cases shows. Moreover, the arguments presented in this paper might be applied to the model studied in \cite{HSYZ} to improve the bounds on the critical line, at least in the high and low temperature regime.

In the following,  we will derive, using a Kramers-Wannier duality argument, quantitative relations between our model and the one with spins on 
the dual graph, which is closely related to the models of  \cite{HSYZ} and \cite{Hern}. 

Let us start by defining the dual graph of a triangulation. Given a finite triangulation $T$, its dual graph $T^*$
 has exactly one vertex on  each triangle face of $T$ and one on the 
 outer face, denoted by $v_O$;  each of 
the edges of  $T^*$ intersects one and only one edge of graph $T$, 
and every edge of $T$ is intersected by one edge of $T^*$. 
Hence, the following relations hold
\begin{align}
	|E(T^*)| & =|E(T)|, \label{eq:edgedual} \\
	|V(T^*)|& =|F(T)|+1, \label{eq:vxdual}
\end{align}
and the degree of every vertex in $T^*$, except $v_O$, is 3. 

Consider now the partition function
\begin{equation}
	Z_{N,l}(\beta,\mu) = \sum_{ T \in \mathcal{T}_{N,l}} e^{-\mu |F(T)|} \sum_{\sigma \in \{-1,1\}^{|V(T)|}} e^{- \beta  H(T, \sigma)}.
\label{eq:Z_fixb}
\end{equation}
The notion of dual graph allows us to express eq. (\ref{eq:Z_fixb}) as 
\begin{equation}
	Z_{N,l}(\beta,\mu) = \sum_{ T \in \mathcal{T}_{N,{k}}} e^{-\mu |F(T)|+ \beta|E(T)|} \, 2 \sum_{\Gamma \subset E(T^*)} e^{- 2 \beta  |\Gamma |},
\label{Zn}
\end{equation}
where $\Gamma$ is a subset of edges (including the empty one)  consisting of  edge-disjoint  connected components of $T^*$ each of which is a cycle. 
In other words, every vertex in $\Gamma$ has an even degree in $T^*$. Such subgraphs are called even subgraphs. The above expression of the partition function is known as low-temperature expansion. 

Observe that the interactions in (\ref{Htri}) through  the adjacent triangles can be rewritten as interaction through the edges of the dual graph. Therefore, the partition function for causal triangulations with spins  on the faces is given by
\begin{equation}
Z_{N,l}^{\Delta}(\beta^*,\mu^*) = \sum_{ T \in \mathcal{T}_{N,l}} e^{-\mu^* |F(T)|} \sum_{\sigma \in \{-1,1\}^{|V(T^*)|}} e^{- \beta^*  H(T^*, \sigma)} ,
\end{equation}
which still slightly differs from the one studied in \cite{HSYZ} and \cite{Hern} (because, as said above, we consider triangulations of a plane, not of a torus). 
Let us  rewrite $Z_{N}^{\Delta}(\beta^*,\mu^*)$ using the high-temperature expansions (\ref{HT}):
\begin{equation}\label{eq:zhte}
	Z_{N,l}^{\Delta}(\beta^*,\mu^*) = \sum_{ T \in \mathcal{T}_{N,l}} e^{-\mu^* |F(T)|}  (\cosh \beta^*)^{|E(T^*)|} 2^{|V(T^*)|} \sum_{\Gamma \subset E(T^*)} (\tanh \beta^*)^{|\Gamma|},
\end{equation} 
where the sum runs over even subgraphs $\Gamma$ as in (\ref{Zn}).

For each $\beta>0$ let us now define the transformations $\beta^*=\beta^*(\beta)$ and $\mu ^*=\mu ^*(\mu, \beta )$, so that
\begin{gather}
	\tanh \beta^* = e^{- 2 \beta} , \label{eq:tempdual} \\
	\mu ^* = \mu + \log 2 -\frac{3}{4}\log{(2 \sinh{2\beta})} \label{eq:cosmdual}.
\end{gather}
Using eq. (\ref{eq:tempdual}) and eqs. (\ref{eq:edgedual})-(\ref{eq:vxdual}),  the partition function (\ref{eq:zhte}) reads
\begin{equation}
\begin{split}
	Z_{N,l}^{\Delta}(\beta^*,\mu^*) & = \sum_{ T \in \mathcal{T}_{N,l}} e^{-\mu^* |F(T)|}  (\cosh \beta^*)^{|E(T)|} 2^{|F(T)|} \left( e^{\beta} (\tanh \beta^*)^{1/2} \right)^{|E(T)|} \\
	& \times 2  \sum_{\Gamma         \subset E(T^*)} e^{-2 \beta |\Gamma|}.
\end{split}
\end{equation} 
Then,  taking into account that
\begin{equation}
	|E(T)|=\frac{3}{2}|F(T)|+ \frac{1}{2}|S_N|=\frac{3}{2}|F(T)|+ \frac{1}{2}l,
\end{equation} 
and noting that eq. (\ref{eq:tempdual}) is equivalent to $\sinh 2\beta \sinh 2\beta^* = 1$, using eq. (\ref{eq:cosmdual}) a straightforward calculation leads to
\begin{equation}
	Z_{N,l}^{\Delta}(\beta^*,\mu^*) = \left(2\sinh{2\beta} \right)^{ -\frac{l}{4}} 	Z_{N,l}(\beta,\mu).
\end{equation} 

\section{Magnetization of the central spin}

First, observe that 
\[Z_{N,l}^+(\beta,\mu)<Z_{N+1}(\beta,\mu).\]
Hence,  the partition function $Z_{N,l}^+(\beta,\mu)$ exists for all $N$ when
$(\mu, \beta ) \in \Delta_{f}$ defined in (\ref{eq:sub_reg}), i.e., 
\begin{equation}
	\mu > \beta + \log(1+ 2 \cosh\beta).
\label{J30}
\end{equation}

For $(l,+)$-boundary conditions, the high-temperature expansion (\ref{HT}) of the partition function reads as
\begin{equation}
	Z_{N,l}^+(\beta,\mu) = 2 (e^{\beta -\mu} \cosh \beta)^l \sum_{\textbf{k}:k_N=l} \sum_{ T \in  \mathcal{T}_{N,\textbf{k}}} \left(\prod_{i=1}^{N-1} (2 e^{-2 \mu} (\cosh \beta)^3)^{k_i} \right) 
\label{eq:Zpl_hte}
\end{equation}
\[\times  \sum_{A \in \mathcal{E}^+(T)} (\tanh \beta)^{|E|},\]
where
\begin{equation}
	\mathcal{E}^+(T) = \{ A \subseteq E(T) \setminus E(S_N) : I(i,A) \, \text{is even} \, \forall i \in  V(T) \setminus V(S_N)\}.
\end{equation}
Similarly, the expected value of the magnetization of the central spin (\ref{eq:s0p_expect}) can be written as
\begin{equation}
	\begin{split}
	\langle \sigma_0 \rangle_{N,l,\beta,\mu}^+ = & \frac{(e^{\beta -\mu} \cosh \beta)^l}{Z_{N,l}^+(\beta,\mu) } \sum_{\textbf{k}:k_N=l} \sum_{ T \in  \mathcal{T}_{N,\textbf{k}}}  \left(\prod_{i=1}^{N-1} (2 e^{-2 \mu} (\cosh \beta)^3)^{k_i} \right) \\
	& \times  \sum_{A \in \mathcal{E}^+_0(T)} (\tanh \beta)^{|E|},
	\end{split}
\label{J31}
\end{equation}
where now
\begin{equation}
	\begin{split}
		\mathcal{E}_0^+(T) = & \{ A\subseteq E(T) \setminus E(S_N) : \\
		& I(i,A) \, \text{is even} \, \forall i \in  V(T)  \setminus \{ v_0\} \setminus V(S_N), \ \ I(v_0,A)  \, \text{is odd} \}.
	\end{split}
\end{equation}

\subsection{Magnetization for finite random triangulations}

\begin{proof}[Proof of Theorem \ref{TC}]
Since for any $T \in \mathcal{T}_{N,l}$ there exists a path $\gamma$, with $|\gamma| = N$, connecting the root $v_0$ to $S_N$, such that $E(\gamma) \in \mathcal{E}_0^+(T)$, we have by (\ref{J31}) 
\begin{equation}
	\begin{split}
		\langle \sigma_0 \rangle_{N,l,\beta,\mu}^+ & \geq \frac{(e^{\beta -\mu} \cosh \beta)^l}{Z_{N,l}^+(\beta,\mu) } (\tanh \beta)^N \sum_{T \in \mathcal{T}_{N,l}} \left(\prod_{i=1}^{N-1} (2 e^{-2 \mu} (\cosh \beta)^3)^{k_i} \right) \\
		& = \frac{(e^{\beta -\mu} \cosh \beta)^l}{Z_{N,l}^+(\beta,\mu) } (\tanh \beta)^N W_{N,l}(2 e^{-2 \mu} (\cosh \beta)^3) > 0
	\end{split}
\end{equation}
for any $(\beta ,\mu)$ which satisfy (\ref{J30}). This (together with the symmetry in the model) proves Theorem 
\ref{TC}. \qed
\end{proof}

\subsubsection{Magnetization at high-temperature. Infinite random triangulation}

\begin{proof}[Proof of Theorem \ref{thm:mag}]
Note that under assumption that $2 e^{-2 \mu} (\cosh \beta)^3 < 1/4$ from eq. (\ref{eq:Zpl_hte}) it follows that
\begin{equation}
	Z_{N,l}^+(\beta,\mu) \geq 2 (e^{\beta -\mu} \cosh \beta)^l W_{N,l}(2 e^{-2 \mu} (\cosh \beta)^3).
\label{eq:Zpl_lobound_ht}
\end{equation}

Consider now formula  (\ref{J31}). Let  $A \in \mathcal{E}^+_0(T)$ be a subset of edges.
Since by the definition  the degree of $v_0$ in $A$
is odd, it has to be positive.
 Hence, the set $A$ is not empty: it contains at least one edge incident to $v_0$.  Consider now the connected component of $A$ which contains $v_0$, denote it $C$. Notice, that $C\subseteq A$. Let also $V(A)$ and $V(C)$ denote the set of vertices spanned by the graphs $A$ and $C$, correspondingly. 
 For any $v\in V(C)$ let $\nu_C(v)$ be the degree of the vertex $v$ in the component $C$. Observe that $\nu_C(v)=I(v, A)$, as we introduced above the incidence number.

Notice that by the property of a graph the sum of all degrees is twice the number of the edges, i.e., 
\begin{equation}
	2|C|=\sum_{v\in V(C)}\nu_C(v) = I(v_0, A) + \sum_{v\in V(C):v\neq v_0}I(v, A).
\end{equation}
Since $I(v_0, A)$ is an odd number, we must have at least one odd number in the remaining sum. But the only vertices in $V(T)$ which might have an odd degree are on the slice $S(N)$. Hence, the connected 
component $C$ has at least one vertex on the slice $S_N$, which implies that $C$
must contain a path between $v _0$ and $S_N$. This gives us a bound 
$|A| \geq |C| \geq N$. Now taking into account that 
$\tanh \beta < 1$,
 we derive from  (\ref{J31})  the inequality
\begin{equation}
	\begin{split}
		\langle \sigma_0 \rangle_{N,l,\beta,\mu}^+ & \leq \frac{(e^{\beta -\mu} \cosh \beta)^l}{Z_{N,l}^+(\beta,\mu) } (\tanh \beta)^N  \\
		 & \times \sum_{\textbf{k}:k_N=l} \sum_{ T \in  \mathcal{T}_{N,\textbf{k}}} \left(\prod_{i=1}^{N-1} (2 e^{-2 \mu} (\cosh \beta)^3)^{k_i} \right) |\mathcal{E}^+_0(T)|.
	\end{split}
\label{eq:mag_bound1}
\end{equation}

Using a rough bound  $|\mathcal{E}^+_0(T)| \leq 2^{|E(T)|-l}$ and  (\ref{eq:Zpl_lobound_ht}), 
we get from (\ref{eq:mag_bound1})
\begin{equation}
	\begin{split}
		\langle \sigma_0 \rangle_{N,l,\beta,\mu}^+ & \leq \frac{(e^{\beta -\mu} \cosh \beta)^l}{Z_{N,l}^+(\beta,\mu)} (\tanh \beta)^N \\
		& \times \sum_{\textbf{k}:k_N=l} \sum_{ T \in  \mathcal{T}_{N,\textbf{k}}} \left(\prod_{i=1}^{N-1} (2 e^{-2 \mu} (\cosh \beta)^3)^{k_i} \right) 2^{|E(T)|-l} \\
		&  =  \frac{(2e^{\beta -\mu} \cosh \beta)^l}{Z_{N,l}^+(\beta,\mu)} (\tanh \beta)^N \sum_{\textbf{k}:k_N=l} \sum_{ T \in  \mathcal{T}_{N,\textbf{k}}}  \prod_{i=1}^{N-1} (16 e^{-2 \mu} (\cosh \beta)^3)^{k_i}  \\
		&  \leq 2^{l-1} (\tanh \beta)^N \frac{W_{N,l}(16 e^{-2 \mu} (\cosh \beta)^3)}{W_{N,l}(2 e^{-2 \mu} (\cosh \beta)^3)}
	\end{split}
\end{equation}
which holds at least for all $\mu > 3/2 \log (\cosh \beta) + 3 \log 2$. 
Therefore, using (\ref{eq:W_asym}), we have that for $\beta$ small enough,
\begin{equation}
	\lim_{N \to \infty}\langle \sigma_0 \rangle_{N,l,\beta,\mu}^+ = 0.
\end{equation}
This proves the statement of Theorem \ref{thm:mag}. \qed
\end{proof}

\section{Proof of Lemma \ref{Mp} }

By the definitions  (\ref{J7})  and (\ref{Gf}) we have
\begin{equation}
	\begin{split}
		p_{N,\beta, \mu} \left\{ (T_K, \sigma_K) \mid  ({\widetilde T},  {\widetilde \sigma}) \right\} & = 
		\frac{ \sum_{ T\in T_K \cup {\widetilde T} }\ p_{N,\beta, \mu} (T, (\sigma_K,  {\widetilde \sigma}))}{\sum_{  (T'_K, \sigma'_K ) \in {\Lambda}_K} \sum_{T\in  T'_K \cup {\widetilde T}}
 p_{N,\beta, \mu} (T, (\sigma'_K,  {\widetilde \sigma})) } \\
		& = \frac{ \sum_{ T\in T_K \cup {\widetilde T} }\ e^{-\beta H(T, (\sigma_K,  {\widetilde \sigma}))-\mu F(T)} }{ \sum_{  (T'_K, \sigma'_K ) \in {\Lambda}_K } 
\sum_{ T\in  T'_K \cup {\widetilde T} } e^{-\beta H(T, (\sigma'_K,  {\widetilde \sigma}))-\mu F(T)} }.
	\end{split}
\label{J1}
\end{equation}

Consider the numerator of the last fraction. 
For any $T\in T_K \cup {\widetilde T}$
let us denote $I(T_K, {\widetilde T})=I (S_K, {\widetilde S}_{K+1})$ the set of edges between these two graphs  $T_K$ and ${\widetilde T}$; call it an 
{\it interaction} set. 
Then  by the definition (\ref{H}) we have for all $T\in T_K \cup {\widetilde T}$
\begin{equation}\label{J2}
H(T, (\sigma_K,  {\widetilde \sigma}))=
H(T_K, \sigma_K) + H( {\widetilde T}, {\widetilde \sigma}) +
H\left( I (S_K, {\widetilde S}_{K+1}), (\sigma_{S_K}, {\widetilde \sigma}_{{\widetilde S}_{K+1}})\right).
\end{equation}
Also, counting the number of triangles we get for all $T\in T_K \cup {\widetilde T}$
\begin{equation}\label{J3}
F(T)=F(T_K)+ F( {\widetilde T})+F(I (S_K, {\widetilde S}_{K+1}))=
F(T_K)+ F( {\widetilde T})+|S_K|+|{\widetilde S}_{K+1}|.
\end{equation}
Similar relations to (\ref{J2}) and (\ref{J3}) hold as well for any $T\in T'_K \cup {\widetilde T}$. Making use of (\ref{J2}) and (\ref{J3}) we derive from (\ref{J1}) 
\begin{equation}\label{J9}
  p_{N,\beta, \mu} \left\{ (T_K, \sigma_K) \mid  ({\widetilde T},  {\widetilde \sigma}) \right\} = \frac{ \sum_{ T\in T_K \cup {\widetilde S}_{K+1} }\ 
e^{-\beta H(T, (\sigma_K,  {\widetilde \sigma}_{K+1}))-\mu F(T)}
}{
\sum_{  (T'_K, \sigma'_K ) \in {\Lambda}_K   }
\sum_{
T\in  T'_K \cup {\widetilde S}_{K+1}  \ 
}
 e^{-\beta H(T, (\sigma'_K,  {\widetilde \sigma}_{K+1}))-\mu F(T)}
}.
\end{equation}
It remains to notice  that the denumerator in  (\ref{J9})  equals 
\begin{equation}\label{J8}
\sum_{ (T, \sigma) \in \Lambda_{K+1, |{\widetilde S}_{K+1}|}^{{\widetilde \sigma}_{K+1}}}
e^{-\beta H(T, \sigma)-\mu F(T)}
=Z_{K+1, |{\widetilde S}_{K+1}|}^{ {\widetilde\sigma}_{K+1}}(\beta, \mu).
\end{equation}

Combining  (\ref{J9}) with  (\ref{J8}) we get the first equality in   (\ref{p1}); the second equality follows by Definition \ref{DBC}. \qed

\appendix

\section{Proof of Lemma \ref{LA}.}
\label{apx:series}

Consider the multiple series
\begin{equation}
	W_{n+1,l}(x) = \sum_{k_1 = 1}^\infty \cdots \sum_{k_n = 1}^\infty \prod_{i=1}^{n} \binom{k_{i+1}+k_i-1}{k_i-1} x^{k_i},
\end{equation}
with $k_{n+1}=l$. Summing over $k_1$ we obtain
\begin{equation}
	\sum_{k_1 = 1}^\infty \binom{k_1+k_2-1}{k_1 -1} x^{k_1} = \frac{x}{1-x} \left( \frac{1}{1-x} \right)^{k_2} = x B_1^{k_2+1},
\end{equation}
where we denoted $B_1 = (1- x)^{-1}$. Inserting it in the equation and summing over $k_2$ we obtain
\begin{equation}
	\sum_{k_2 = 1}^\infty\binom{k_2+k_3-1}{k_2 -1}  x B_1  (x B_1)^{k_2} = \frac{x^2 B_1^2}{1 - x B1} \left(\frac{1}{1 - x B_1}\right)^{k_3} = x^2 B_1^2 B_2^{k_3+1},
\end{equation}
where $B_2 = (1- x B_1)^{-1}$. Summing over  the remaining $k_i$'s, we obtain
\begin{equation}
	W_{n+1,l}(x) = x^n B_n(x)^{l+1}\prod_{i=1}^{n-1} B_i(x)^2,
\label{eq:W_sol}
\end{equation}
where $B_i(x)$ is the solution to  the recursion relation
\begin{equation}
\begin{split}
	B_i & = \frac{1}{1-x B_{i-1}}, \\
	B_1 & = \frac{1}{1-x},
\end{split}
\end{equation}
This reads
\begin{equation}
	B_i(x) = 2 \frac{c_+(x)^{i+1} - c_-(x)^{i+1}}{c_+(x)^{i+2} - c_-(x)^{i+2}},
\label{eq:rec_sol}
\end{equation}
with 
\begin{equation}
	c_\pm(x) = 1 \pm \sqrt{1-4x}.
\end{equation}

For $0<x<1/4$, substituting (\ref{eq:rec_sol}) into (\ref{eq:W_sol}) we get
\begin{equation}
	W_{n+1,l}(x) = 2^{l+3} (1-4x) (4x)^n \frac{(c_+(x)^{n+1} - c_-(x)^{n+1})^{l-1}}{(c_+(x)^{n+2} - c_-(x)^{n+2})^{l+1}}, 
\end{equation}
which gives
\begin{equation}
	\lim_{n \to \infty} W_{n,l}(x) = 0.
\end{equation}
In particular we have
\begin{equation}
	W_{n+1,l}(x) \sim f_{l}(x) \left( \frac{4x}{c_+(x)^2} \right)^{n}, \quad \text{for } n \to \infty, 
\end{equation}
where
\begin{equation}
	f_l(x) = (1-4x) \left( \frac{2}{c_+(x)} \right)^{l+3} .
\end{equation}
This yields the first statement of Lemma \ref{LA}. The rest follows directly by the results of \cite{MYZ01}. \qed

\section{Proof of Corollary \ref{CA}.}

We have that, for any $x \in (0,1/4)$, $B_i(x)$ is monotonically increasing
\begin{equation}
	\begin{split}
		\frac{B_{i-1}(x)}{B_{i}(x)} & = \frac{(c_+(x)^{n+2} - c_-(x)^{n+2})(c_+(x)^{n} - c_-(x)^{n})}{(c_+(x)^{n+1} - c_-(x)^{n+1})^2} \\
		& = \frac{c_+(x)^{n+1} + c_-(x)^{n+1} - (c_-(x) c_+(x))^{n} (c_+^2 + c_-^2)}{(c_+(x)^{n+1} - c_-(x)^{n+1})^2} \\
		& = 1 - \frac{(c_-(x) c_+(x))^{n} (c_+ - c_-)^2}{(c_+(x)^{n+1} - c_-(x)^{n+1})^2} \\
		& = 1 -  \frac{4 (4x)^n (1-4x)}{(c_+(x)^{n+1} - c_-(x)^{n+1})^2} < 1.
	\end{split}
\end{equation}
Therefore, using that
\begin{equation}
	\lim_{i \to \infty} B_i(x) = \frac{2}{c_+(x)},
\end{equation}
we obtain that, for any $i \in \mathbb{N}$ and $x \in (0,1/4)$,
\begin{equation}
	B_i(x) < \frac{2}{1 + \sqrt{1-4x}}.
\label{eq:B_bound}
\end{equation}

Now, consider the series
\begin{equation}
	W_{n+1}(x,y) = \sum_{l \geq 1} y^l W_{n+1,l}(x).
\end{equation}
By eq. (\ref{eq:W_sol}), we have that the series is convergent if and only if
\begin{equation}
	y B_{n}(x) < 1. 
\label{I5}
\end{equation}
Formula (\ref{eq:B_bound}) together with Lemma \ref{LA} yield that  the inequality 
(\ref{I5}) 
is satisfied for any $n \in \mathbb{N}$ and  if and only if
$x \in (0,1/4]$ and
\begin{equation}
	\frac{2 y}{1 + \sqrt{1-4x}} <1,
\end{equation}
that is
\begin{equation}
	y^2 - y + x < 0.
\end{equation}
This proves the corollary. \qed

\section*{acknowledgements} 
The authors thank Prof. Bergfinnur Durhuus for the helpful discussions on the subject. This work was partially supported by the Swedish Research Council through the grant no. 2011-5507.

\end{document}